\documentclass[12pt, a4paper]{article}

\usepackage{latexsym}
\usepackage{mathrsfs}
\usepackage{amsmath}
\usepackage{amsthm}
\usepackage{array}
\newtheorem{theorem}{Theorem}

\newtheorem{lemma}{Lemma}

\newtheorem{corollary}{Corollary}

\usepackage{pst-all}
\usepackage{mathrsfs}

\begin{document}

\title{Topological properties on  the diameters of the integer simplex}

\author
{{\large Meijie Ma}  \\
{\small Department of Mathematics }   \\
{\small Zhejiang Normal University, Jinhua 321004, China}  \\
}
\date{}
\maketitle

\setlength{\baselineskip}{24pt}

\begin{abstract}
Wide diameter $d_\omega(G)$ and fault-diameter $D_\omega(G)$ of an
interconnection network $G$ have been recently studied by many
authors. We determine the wide diameter and fault-diameter of the
integer simplex $T_m^n$. Note that $d_1(T_m^n)=D_1(T_m^n)=
d(T_m^n)$, where $d(T_m^n)$ is the diameter of $T_m^n$. We prove
that $d_\omega(T_m^n)=D_\omega(T_m^n)= d(T_m^n)+1$ when
$2\leq\omega\leq n$. Since a triangular pyramid $TP_L$ is $T_L^3$,
we have $d_\omega(TP_L)=D_\omega(TP_L)= d(TP_L)+1$ when
$2\leq\omega\leq 3$.
\end{abstract}


\section{Introduction}

An interconnection network is conveniently represented by an
undirected graph. The vertices(edges) of the graph represent the
nodes(links) of the network. As a topology for an interconnection
network of a multiprocessor system, the triangular pyramid (tripy
for short), is proposed by Razavi and Sarbazi-Azad in~\cite{rs10}.
Some basic properties such as Hamiltonian-connectivity, pancyclicity
and a routing algorithm were investigated in the paper. We studied
other properties such as symmetry, connectivity and fault-tolerant
vertex-pancyclicity in~\cite{mw13}.

Reliability and efficiency are important criteria in the design of
interconnection networks. In graph theory and the study of
fault-tolerance and transmission delay of networks, wide diameter
and fault-diameter are two very important parameters and have been
studied by many researchers. The diameters of  Cartesian product
graphs were studied in~\cite{bez09, bz08, bz10, ez13, ez15, xxh05}.
The parameters of some well-known networks such as hypercube,
crossed cube etc. were studied in~\cite{cw09, c10, k15, kllht09,
lcch98, sl12, ylcz05}.

Wide diameter of a graph, which combines connectivity with diameter,
is a parameter that measures simultaneously the fault-tolerance and
efficiency of parallel processing computer networks. Let $u$ and $v$
be two vertices of a graph $G$. The distance between $u$ and $v$,
denoted by $d_G(u,v)$, is the length of the shortest path between
them. The diameter of $G$, denoted by $d(G)$, is the maximum
distance between any two vertices. The connectivity $\kappa(G)$ of a
graph $G$ is the minimum number of vertices whose removal results in
a disconnected or trivial network. We say that $G$ is $k$-connected
for any $0 < k \leq \kappa(G)$. According to Menger's theorem, there
are $k$ disjoint paths between any two vertices in a $k$-connected
network. Let $G$ be a $k$-connected graph with $1\leq \omega \leq k
$. The $\omega$-wide diameter $d_\omega(G)$ of $G$ is the minimum
$\ell$ such that there exist $\omega$ internally vertex disjoint
paths of length at most $\ell$ from $u$ to $v$ for any two vertices
$u$ and $v$. Throughout this paper, ``disjoint paths" always means
``internally vertex disjoint paths". In particular, $d_1(G)$ is just
the diameter $d(G)$ of $G$. It is easy to see that
$$d(G)=d_1(G)\leq d_2(G)\leq \cdots \leq d_{k-1}(G)\leq d_k(G).
$$

Failures are inevitable when a network is put in use. Therefore, it
is significant to consider faulty networks. The fault-diameter can
be used to estimate the effect on the diameter when faults occur. A
small fault-diameter is desirable because the delay would be shorter
when some nodes fail. The fault-diameter with faulty vertices was
first studied by the authors in~\cite{kk87}. The
$(\omega-1)$-fault-diameter of a graph $G$ is defined as
$$D_\omega(G)=\max\{d(G-F): F\subseteq V(G), |F|<\omega\}$$
where $G-F$ denotes the subgraph induced by $V(G)-F$. Note that
$D_\omega(G)<\infty$ if and only if $G$ is $\omega$-connected. It is
also clear that
$$d(G)=D_1(G)\leq D_2(G)\leq \cdots \leq
D_{k-1}(G)\leq D_k(G). $$

It is well known (see~\cite{lcch98}) that for any $k$-connected
graph $G$ and any integer $\omega, 1\leq \omega \leq k$, we have
$D_\omega(G)\leq d_\omega(G)$, where the equality holds for some
well-known networks. However, it's difficult to determine the wide
diameter or fault-diameter of the tripy according to its definition
in~\cite{rs10}. Fortunately, we find that the tripy $TP_L$ is a
special integer simplex, and we determine the wide diameter and
fault-diameter of the integer simplex. The two kinds of diameters of
the tripy are deduced from the results of the integer simplex.

The rest of this paper is organized as follows. We give some
definitions and notations in Section 2. The main results are given
in Section 3.

\section{The Integer Simplex $T_m^n$}

For graph-theoretical terminology and notation not defined here, we
follow~\cite{x01}. We first restate the definitions of triangular
mesh and tripy originally proposed by Razavi and Sarbazi-Azad for
completeness.

\begin{figure}[h]
\begin{center}
\begin{pspicture}(1.75,2.6)(6.5,7)
\psset{radius=.11, xunit=1.5, yunit=1.4}

\Cnode(3.,4.44){00}\rput(3.,4.7){\scriptsize $(0,0)$}
\Cnode(2.5,3.58){10}\rput(2.1,3.58){\scriptsize $(1,0)$}
\Cnode(3.5,3.58){01}\rput(3.85,3.58){\scriptsize $(0,1)$}
\Cnode(2,2.72){20}\rput(1.6,2.72){\scriptsize $(2,0)$}
\Cnode(3,2.72){11}\rput(3.45,2.88){\scriptsize $(1,1)$}
\Cnode(4.0,2.72){02}\rput(4.35,2.72){\scriptsize $(0,2)$}
\Cnode(1.5,1.86){30}\rput(1.1,1.82){\scriptsize $(3,0)$}
\Cnode(2.5,1.86){21}\rput(2.12,2.02){\scriptsize $(2,1)$}
\Cnode(3.5,1.86){12}\rput(3.95,2.02){\scriptsize $(1,2)$}
\Cnode(4.5,1.86){03}\rput(4.85,1.82){\scriptsize $(0,3)$}
%\Cnode(1.0,1){40}\rput(0.7,0.72){\scriptsize $(4,0)$}
%\Cnode(2.0,1){31}\rput(2,0.72){\scriptsize $(3,1)$}
%\Cnode(3.0,1){22}\rput(3,0.72){\scriptsize $(2,2)$}
%$\Cnode(4.0,1){13}\rput(4,0.72){\scriptsize $(1,3)$}
%$\Cnode(5.0,1){04}\rput(5.25,0.72){\scriptsize $(0,4)$}

\ncline{00}{10}\ncline{00}{01}\ncline{01}{10}\ncline{20}{10}
\ncline{01}{02}\ncline{20}{11}\ncline{20}{30}\ncline{10}{11}
\ncline{01}{11}\ncline{02}{11}\ncline{21}{11}\ncline{12}{11}
\ncline{30}{31}\ncline{30}{21}\ncline{21}{12}
\ncline{21}{20}\ncline{02}{12}
\ncline{03}{12}\ncline{03}{02}\ncline{21}{22}\ncline{21}{31}
%\ncline{13}{12}\ncline{22}{12}\ncline{03}{04}\ncline{03}{13}
%\ncline{13}{04}\ncline{13}{22}\ncline{31}{22}\ncline{31}{40}
%\ncline{30}{40}\ncline{21}{22}\ncline{21}{31}
\end{pspicture}
\caption{\label{f1}\footnotesize {The $T_3$.}}
\end{center}
\end{figure}

\noindent {\bf Definition 1}\quad A radix-$m$  triangular mesh
network, denoted by $T_m$, consists of a set of vertices $V(T_m) =
\{(x,y)|0\leq x+y\leq m\}$ where any two vertices $(x_1,y_1)$ and
$(x_2,y_2)$ are connected by an edge if and only if
$|x_1-x_2|+|y_1-y_2|=1$, or $x_2=x_1+1$ and $y_2=y_1-1$, or
$x_2=x_1-1$ and $y_2=y_1+1$.

Figure 1 shows a $T_3$ network.

A tripy is a hierarchy structure based on triangular meshes.

\noindent{\bf Definition 2}\quad An $L$-level tripy, denoted by
$TP_L$, consists of a set of vertices $V(TP_L) = \{(k,(x,y))| 0\leq
k\leq L, 0\leq x + y \leq k\}$. Vertex $(k, (x,y))\in V(TP_L)$ is
said to be a vertex at level $k$ with the coordinate $(x,y)$. The
vertices at level $k$ form a network of $T_k$. Vertex $(k, (x,y))$
is also connected to vertices $(x,y)$, $(x + 1,y)$, and $(x,y + 1)$,
in level $k + 1$, as child vertices, and to vertex $(x-1,y)$ if
$x>0$, vertex $(x,y-1)$ if $y> 0$, and vertex $(x,y)$ if $x+y < k$,
as parents in level $k-1$.

\noindent {\bf Definition 3}\quad The integer simplex with dimension
$n$ and side-length $m$ is the graph $T_m^n$ whose vertices are the
nonnegative integer $(n + 1)$-tuples summing to $m$, with two
vertices adjacent when they differ by $1$ in two places and are
equal in all other places.

In other word,
$$V(T_m^n)=\{v_n\cdots v_1v_0 \, |\, \Sigma_{i=0}^{n}v_i=m,
v_i\, (i\in \{0,1,\ldots,n\})\ \mbox{is a non-negtive integer} \}.$$
Two vertices $u=u_n\cdots u_1u_0$ and $v=v_n\cdots v_1v_0$ are
adjacent if and only if $u_i=v_i-1$, $u_j=v_j+1$ and $u_k=v_k$ where
$k\in \{0,1,\ldots, n\}\setminus \{i,j\}$.

However, both the triangular mesh and the tripy are special kinds of
integer simplex. Let $\sigma_1$ be a mapping from $V(T_m)$ to
$V(T_m^2)$ defined by $\sigma_1 ((x,y))=(m-(x+y),x,y)$, for any
vertex $(x,y)\in V(T_m)$. Then, $\sigma_1$ is an isomorphism from
$T_m$ to $T_m^2$. Let $\sigma_2$ be a mapping from $V(TP_L)$ to
$V(T_L^3)$ defined by $\sigma_2 ((k,(x,y)))=(m-(k+x+y),k,x,y)$, for
any vertex $(k,(x,y))\in V(TP_L)$. Then, $\sigma_2$ is an
isomorphism from $TP_L$ to $T_L^3$.

According to the definition, the number of vertices of $T_m^n$ is
${n+m}\choose {m}$, and the minimum degree of $T_m^n$ is $n$. The
special cases of $T_m^n$ are the following. $T_1^d$ is $K_{d+1}$;
$T_m^1$ is a path with $m+1$ vertices; $T_m^2$ is the triangular
mesh network $T_m$, and is also named as triangulated triangle with
side length $m$ in~\cite{hms95}; $T_m^3$ is the triangular pyramid
$TP_m$ which is studied in~\cite{mw13, rs10}. Since $T_m^1$ is a
path, its connectivity is $1$. We assume $n\geq 2$ when we studied
the wide diameter and fault-diameter of $T_m^n$.

\section{Main results}

We first construct disjoint paths of certain lengths  joining two
vertices in $T_m^n$. A path $P$ joining vertices $u$ and $v$ is also
denoted by a $uv$-path. For any two vertices $u=u_n\cdots u_1u_0$
and $v=v_n\cdots v_1v_0$ of $T_m^n$, let
$h(u,v)=\frac{1}{2}\sum_{i=0}^{n}|u_i-v_i|$. By the vertex definiton
of $T_m^n$, $h(u,v)\leq m$.

\begin{lemma}\label{t1}
Let $u=u_n\cdots u_1u_0$ and $v=v_n\cdots v_1v_0$ be two vertices in
$T_m^n$. If there are $p$ bit positions such that $u_i>v_i$, and $q$
bit positions such that $u_i<v_i$, then there are $n+1-(p+q)+pq$
disjoint $uv$-paths, such that $pq$ of them are of length $h(u,v)$,
and the other $n+1-(p+q)$ of them are of length $h(u,v)+1$.
\end{lemma}

\begin{proof}
Without loss of generality, we may assume $u_i>v_i$ where $i\in\{n,
n-1, \ldots, n-p+1\}$, and $u_i<v_i$ where $i\in\{0, 1, \ldots,
q-1\}$. (Note that $p\geq 1$, $q\geq 1$, and $p+q\leq n+1$). Then,
$h(u,v)=\sum_{i=n-p+1}^{n}(u_i-v_i)=\sum_{i=0}^{q-1}(v_i-u_i)$.

There are $q$ different ways to make $u_{q-1}\cdots u_0$ up to
$v_{q-1}\cdots v_0$. We modify the bits $(q-1),(q-2),\ldots,1,0$ in
a rotational way.

The first way is that we modify the $(q-1)$th bit $u_{q-1}$ up to
$v_{q-1}$, and then modify the $(q-2)$th bit $u_{q-2}$ up to
$v_{q-2}$ and so on. At last modify the $0$th bit $u_0$ up to $v_0$.
We use $(q-1)-(q-2)-\cdots-1-0$ to denote this way. That is
$u_{q-1}u_{q-2}\cdots u_0\rightarrow (u_{q-1}+1)u_{q-2}\cdots u_0
\rightarrow \cdots\rightarrow v_{q-1}u_{q-2}\cdots u_0\rightarrow
v_{q-1}(u_{q-2}+1)\cdots u_0\rightarrow \cdots\rightarrow
v_{q-1}v_{q-2}\cdots v_1u_0\rightarrow v_{q-1}v_{q-2}\cdots
v_1(u_0+1)\rightarrow \cdots\rightarrow v_{q-1}\cdots v_0$.

%The second way is $(q-2)-(q-3)-\cdots-1-0-(q-1)$. That is, we modify
%$(q-2)$th bit $u_{q-2}$ up to $v_{q-2}$, and then modify the
%$(q-3)$th bit $u_{q-3}$ up to $v_{q-3}$ and so forth until the $0$th
%bit is reached, after which we modify the $(q-1)$th bit $u_{q-1}$ up
%to $v_{q-1}$.

In general, the $i$th way is
$(q-i)-(q-i-1)-\cdots-1-0-(q-1)-\cdots-(q-i+1)$ for any $1\leq i\leq
q$.

Similarly, there are $p$ different ways to make $u_n\cdots
u_{n-p+1}$ down to $v_n\cdots v_{n-p+1}$. We modify the bits
$n,(n-1),\ldots,(n+1-p)$ in a rotational way.

The first way is that we modify the $n$th bit $u_n$ down to $v_n$,
and then modify the $(n-1)$th bit $u_{n-1}$ down to $v_{n-1}$ and so
on. At last modify the $(n-p+1)$th bit $u_{n-p+1}$ down to
$v_{n-p+1}$. We use $n-(n-1)-\cdots-(n-p+1)$ to denote this way.

In general, the $j$th way is
$(n+1-j)-(n-j)-\cdots-(n-p+1)-n-(n-1)-\cdots-(n+2-j)$ for any $1\leq
j\leq q$.

Combining an $i$th way to make $u_{q-1}\cdots u_0$ up to
$v_{q-1}\cdots v_0$ and a $j$th way to make $u_n\cdots u_{n-p+1}$
down to $v_n\cdots v_{n-p+1}$, we obtain a $uv$-path of length
$h(u,v)$. There are $pq$ different ways to combine them and these
paths are disjoint.

If $p+q=n+1$, the conclusion is true.

If $p+q< n+1$, then $u_k=v_k<m$ for any bit $k$ where $q\leq k\leq
n-p$. We construct $n+1-(p+q)$ disjoint $uv$-paths of length
$h(u,v)+1$ in the following.

For any $k$ satisfies $q\leq k\leq n-p$, let
$u'=(u_n-1)u_{n-1}\cdots (u_k+1)\cdots u_1u_0$ and
$v'=v_nv_{n-1}\cdots (v_k+1)\cdots v_1(v_0-1)$. We have
$h(u',v')=h(u,v)-1$. We can construct a $u'v'$-path  $P'_k$ of
length $h(u',v')$ as above and the $k$th bit of every vertex on the
path is $v_k+1$. Then $P_k=u+P'_k+v$ is a $uv$-path of length
$h(u,v)+1$.

From the construction, these $(n+1-(p+q))$  $uv$-paths are disjoint
and they are also disjoint with the previous $pq$ $uv$-paths.
\end{proof}

For any two vertices of $T_m^n$, we construct $n+1-(p+q)+pq$
disjoint $uv$-paths. Since $p\geq 1$ and $q\geq 1$, then $pq\geq
(p+q)-1$, and $n+1-(p+q)+pq\geq n$. Hence, we construct at least $n$
disjoint $uv$-paths with length at most $m+1$. By Menger's Theorem,
$n\leq \kappa(T_m^n)$. Note that the minimum degree of $T_m^n$ is
$n$. By the well-known inequality $\kappa(G)\leq \lambda(G)\leq
\delta (G)$, we have the following corollary.

\begin{corollary}\label{l1}
$\kappa(T_m^n)= \lambda(T_m^n)= \delta (T_m^n)=n$ and
$d_n({T_m^n})\leq m+1$.
\end{corollary}

For short we use $a^i$ to denote $i$ identical bit $a$. For example,
$m0^{3} = m000$. Let $u=u_n\cdots u_1u_0$ and $v=v_n\cdots v_1v_0$
be two vertices of $T_m^n$. Assume $u_i>v_i$ where $i\in\{n, n-1,
\ldots, n-p+1\}$, and $u_i<v_i$ where $i\in\{0, 1, \ldots, q-1\}$.
Then from the definition of $T_m^n$, the distance between $u$ and
$v$ is at least $h(u,v)$. By Lemma~\ref{t1}, there is a path of
length $h(u,v)$ joining them. Letting $u=m0^{n-1}$ and $v=0^{n-1}m$,
we have $h(u,v)=m$. The following corollary is obtained.

\begin{corollary}\label{l2}
For any two vertices $u$ and $v$ in $T_m^n$,  $d(u,v)=h(u,v)$.
Furthermore, $d(T_m^n)=m$.
\end{corollary}

\begin{lemma}\label{lem1} Let $u=m0^{n}$ and $v=0^{n}m$ be two vertices  of
$T_m^n$. Any shortest $uv$-path must contain the vertex
$(m-1)0^{n-1}1$.

\end{lemma}

\begin{proof} Since $h(u,v)=m$, the length of a shortest $uv$-path is $m$.

Let $u'$ be a neighbor of $u$.  The leftmost bit of $u'$ is $m-1$,
and the other bits are $0$ except that one bit is $1$. If $u'$ is
not the vertex $(m-1)0^{n-1}1$, then the distance between $u'$ and
$v$ is $m$. Hence, the length of a $uv$-path containing $u'$ is at
least $m+1$. In other word, any shortest $uv$-path must contain the
vertex $(m-1)0^{n-1}1$.
\end{proof}

\begin{lemma}\label{l3}  $D_{2}(T_m^n)\geq m+1$.
\end{lemma}

\begin{proof}
Since the fault-diameter $D_{2}(T_m^n)$ is  defined as the largest
distance between any pair of vertices if a fault occurs and our goal
is to give a lower bound of this. We may select the following three
vertices as follows. Let $u=m0^{n}$ and $v=0^{n}m$ be two vertices
of $T_m^n$, and let $w= (m-1)0^{n-2}1$ be the faulty vertex and
$G=T_m^n-w$. By Lemma~\ref{lem1}, we have $d_{T_m^n}(u,v)=m$ and any
shortest path in $T_m^n$ must contain the vertex $w$. Since
$G=T_m^n-w$, $d_G(u,v) > d_{T_m^n}(u,v)$. We have $d_G(u,v)\geq
m+1$. Hence, $D_2(T_m^n)\geq m+1$.
\end{proof}

We have shown that  $\kappa (T_m^n)=n$. We determine the
fault-diameter $D_\omega(T_m^n)$ and wide diameter $d_\omega(T_m^n)$
for any $2\leq \omega \leq n$ in the following.

\begin{theorem}\label{t2}
$d_\omega(T_m^n)=D_\omega(T_m^n)=d(T_m^n)+1$ for $2 \leq \omega\leq
n$.
\end{theorem}

\begin{proof}
We have $d_n(T_m^n)\leq m+1$ by Corollary~\ref{l1} and $D_2(T_m^n)
\ge m+1$ by Lemma~\ref{l3}. Since $D_2(T_m^n) \le d_2(T_m^n)$,
$D_{n}(T_m^n) \le d_{n}(T_m^n)$, $d_2(T_m^n) \leq \ldots \leq
d_{n}(T_m^n)$ and $D_2(T_m^n)\leq D_3(T_m^n)\leq \ldots \leq
D_{n}(T_m^n)$, we have $d_\omega(T_m^n)=D_\omega(T_m^n)=m+1$ for
$2\leq \omega \leq n$.

By Corollary~\ref{l2}, $d(T_m^n)=m$. The conclusion is true.
\end{proof}

For the tripy, we have the conclusion.

\begin{corollary}\label{c2}
$d_\omega(TP_L)=D_\omega(TP_L)=d(TP_L)+1$ for $2 \leq \omega\leq 3$.
\end{corollary}

Theorem~\ref{t2} shows that the wide diameter equals the
fault-diameter for the integer simplex $T_m^n$. They are increased
only by one over the traditional diameter. Since the triangular
pyramid is a special integer simplex, the tripy enjoys the similar
property of strong resilience like the hypercube. And, the results
add to the attractiveness of the tripy as compared to the hypercube.

\section*{Acknowledgements}

The author would like to thank the support from NSFC (No. 11101378)
and ZJNSF (No. LY14A010009). We also thank Douglas B. West who
introduced the integer simplex to us which makes the definitions of
the triangular mesh and the triangular pyramid more easily and
clearly.

\setlength{\baselineskip}{18pt}


\begin{thebibliography}{s2}

\bibitem{bez09}
I. Bani\v c, R. Erve\v s, J. \v Zerovnik, The edge fault-diameter of
Cartesian graph bundles, European Journal of Combinatorics 30 (2009)
1054-1061.

\bibitem{bz08}
I. Bani\v c, J. \v Zerovnik, The fault-diameter of Cartesian
products, Advances in Applied Mathematics 40 (2008) 98-106.

\bibitem{bz10}
I. Bani\v c, J. \v Zerovnik, Wide diameter of Cartesian graph
bundles, Discrete Mathematics 310 (2010) 1697-1701.



\bibitem{cw09}
C.-P. Chang, C.-C. Wu, Conditional fault diameter of crossed cubes,
Journal of Parallel and Distributed Computing 69 (2009) 91-99.


\bibitem{c10}
X.-B. Chen, Edge-fault-tolerant diameter and bipanconnectivity of
hypercubes, Information Processing Letters 110 (2010) 1088-1092.




%\bibitem{dhl96}
%D.Z. Du, D.F. Hsu, Y.D. Lyuu, On the diameter vulnerability of Kautz
%digraphs, Discrete Mathematics 151 (1996) 81-85.

\bibitem{ez13}
R. Erve\v s, J. \v Zerovnik, Mixed fault diameter of Cartesian graph
bundles, Discrete Applied Mathematics 161 (2013) 1726-1733.


\bibitem{ez15}
R. Erve\v s, J. \v Zerovnik, Improved upper bounds for vertex and
edge fault diameters of Cartesian graph bundles, Discrete Applied
Mathematics 181 (2015) 90-97.



\bibitem{hms95}
R. Hochberg, C. McDiarmid, M. Saks, On the bandwidth of triangulated
triangles, Discrete Mathematics 138 (1995) 261-265.


%\bibitem{hl94}
%D.F. Hsu, T. Luczak, Note on $k$-diameter of $k$-regular
%$k$-connected graphs, Discrete Mathematics 132 (1994) 291-296.

\bibitem{k15}
B. Kraft, Diameters of Cayley graphs generated by transposition
trees, Discrete Applied Mathematics (2014),
http://dx.doi.org/10.1016/j.dam.2014.10.019.

\bibitem{kk87}
M.S. Krishnamoorthy, B. Krishnamurthy, Fault diameter of
interconnection networks, Computers and Mathematics with
Applications 13 (1987) 577-582.

\bibitem{kllht09}
T.-L. Kung, C.-K. Lin, T. Liang, L.-Y. Hsu, J.J.M. Tan, Fault
diameter of hypercubes with hybrid node and link faults, Journal of
Interconnection Networks 10 (3) (2009) 233-242.


%\bibitem{lc99}
%S.C. Liaw, G.J. Chang, Wide diameters of butterfly networks,
%Taiwanese Journal of Mathematics 3 (1) (1999) 83-88.

\bibitem{lcch98}
S.C. Liaw, G.J. Chang, F. Cao, D.F. Hsu, Fault-tolerant routing in
circulant networks and cycle prefix networks, Annals of
Combinatorics 2 (1998) 165-172.




\bibitem{mw13}
M. Ma, P. Wang, Some new topological properties of the triangular
pyramid networks, Information Sciences 248 (2013) 239-245.



\bibitem{rs10}
S. Razavi, H. Sarbazi-Azad, The triangular pyramid: Routing and
topological properties, Information Sciences 180 (2010) 2328-2339.





\bibitem{sl12}
T. Shi, M. Lu, Fault-tolerant diameter for three family
interconnection networks, Journal of Combinatorial Optimization 23
(2012) 471-482.



\bibitem{xxh05}
M. Xu, J.-M. Xu, X.-M. Hou, Fault diameter of Cartesian product
graphs, Information Processing Letters 93 (2005) 245-248.

\bibitem{x01}
J.-M. Xu,  Topological Structure and Analysis of Interconnection
Networks, Kluwer Academic Publishers, Dordrecht/Boston/London, 2001.

\bibitem{ylcz05}
J.H. Yin, J.S. Li, G.L. Chen, C. Zhong, On the fault-tolerant
diameter and wide diameter of $\omega$-connected graphs, Networks 45
(2005) 88-94.



\end{thebibliography}
\end{document}